\documentclass[11pt]{article}
\usepackage{etex}
\usepackage[utf8]{inputenc}
\usepackage[TS1,T1]{fontenc}
\usepackage{textcomp}
\usepackage[a4paper,tmargin=3cm,bmargin=3cm,rmargin=2.2cm,lmargin=2.2cm]{geometry}
\usepackage{lmodern}

\usepackage{amsfonts,amssymb,amsmath,mathtools,cite,slashed}
\usepackage[amsthm,thmmarks,amsmath]{ntheorem}

\usepackage{titlesec}

\DeclareFontFamily{U}{mathx}{\hyphenchar\font45}
\DeclareFontShape{U}{mathx}{m}{n}{
     <5> <6> <7> <8> <9> <10>
     <10.95> <12> <14.4> <17.28> <20.74> <24.88>
     mathx10
     }{}
\DeclareSymbolFont{mathx}{U}{mathx}{m}{n}
\DeclareFontSubstitution{U}{mathx}{m}{n}
\DeclareMathAccent{\widecheck}    {0}{mathx}{"71}

\usepackage{graphicx}
\usepackage[all]{xy}
\usepackage{lmodern}
\usepackage[french,english]{babel}
\usepackage[pdftex]{hyperref}

\usepackage{tikz}
\usetikzlibrary{shapes.misc,arrows,decorations.markings}

\DeclareMathOperator{\Dom}{Dom}      
\DeclareMathOperator{\Tr}{Tr}                 

\newtheorem{assumption}{Assumption}[section]
\newtheorem{theorem}[assumption]{Theorem}

\newtheorem{lemma}[assumption]{Lemma}

\newtheorem{definition}[assumption]{Definition}

\newtheorem{remark}[assumption]{Remark}

\newcommand{\A}{\mathcal{A}}               
\renewcommand{\a}{\alpha}                     
\renewcommand{\b}{\beta}                        
\newcommand{\B}{\mathcal{B}}               
\newcommand{\C}{\mathbb{C}}               
\newcommand{\CC}{\mathcal{C}}            
\newcommand{\del}{\partial}                     
\newcommand{\Det}{\text{Det}}                 
\newcommand{\DD}{\mathcal{D}}            
\newcommand{\eps}{\varepsilon}             
\newcommand{\ga}{\gamma}                    
\renewcommand{\H}{\mathcal{H}}            
\newcommand{\half}{{\mathchoice{\thalf}{\thalf}{\shalf}{\shalf}}}
\renewcommand{\L}{\mathcal{L}}              
\newcommand{\<}{\langle}

\newcommand{\M}{\mathcal{M}}              
\newcommand{\N}{\mathbb{N}}               
\newcommand{\OO}{\mathcal{O}}            
\newcommand{\R}{\mathbb{R}}                
\newcommand{\RR}{\mathcal{R}}                

\newcommand{\set}[1]{\{\,#1\,\}}               
\newcommand{\shalf}{{\scriptstyle\frac{1}{2}}} 
\renewcommand{\SS}{\mathcal{S}}         
\newcommand{\thalf}{\tfrac{1}{2}}             
\newcommand{\tr}{\text{tr}}                         
\newcommand{\wt}{\widetilde}                  
\newcommand{\vc}{\vcentcolon =}             

\def\ee_#1{e_{{\scriptscriptstyle#1}}}       
\def\<#1,#2>{\langle#1\,,\,#2\rangle}          
\newcommand{\norm}[1]{\left\lVert#1\right\rVert}    

\newcommand{\be}{\begin{enumerate}}

\def\XXint#1#2#3{{\setbox0=\hbox{$#1{#2#3}{\int}$}
	\vcenter{\hbox{$#2#3$}}\kern-0.5\wd0}}

\begin{document}

\thispagestyle{empty}

\begin{center}
%

\vspace{3cm}

{\Large\textbf{Spectral action beyond the weak-field approximation}} \\
\vspace{0.5cm}

{\large  B. Iochum$^{1}$, C. Levy$^{2}$, D. Vassilevich$^{3,4}$}

\vspace{3cm}

\abstract{
The spectral action for a non-compact commutative spectral triple is computed covariantly in a gauge perturbation 
up to order 2 in full generality. In the ultraviolet regime, $p\to\infty$, the action decays as $1/p^4$ in any even 
dimension. 
}
\end{center}

\vspace{2cm}

CPH-SYM-00
\vspace{7cm}

{\small
\noindent $^1$ UMR 6207

-- Unit\'e Mixte de Recherche du CNRS, Aix-Marseille Université et Universit\'e du Sud Toulon-Var 

-- Laboratoire affili\'e \`a la FRUMAM -- FR 7332\\
$^2$ Department of Mathematical Sciences, University of Copenhagen, Denmark, levy@math.ku.dk \\
$^3$ CMCC -- Universidade Federal do ABC, Santo Andr\'e, S.P., Brazil, dmitry.vasilevich@ufabc.edu.br\\
$^4$ Also at the Department of Theoretical Physics, St.~Petersburg University, Russia
}

\newpage

\section{Introduction}
Recent advances \cite{CCM,ConnesMarcolli} in explaining some key features
of gravity and Standard Model through the spectral action of noncommutative
geometry brought this subject to a focus of interest in theoretical physics. 
In noncommutative geometry, all information is encoded in a spectral triple $(\A,\H,\DD)$, where $\A$ is an algebra 
acting on a Hilbert space $\H$ and $\DD$ is a selfadjoint operator on $\H$ which plays the role of a Dirac operator 
\cite{Book, ConnesMarcolli, Polaris}. In this approach, the action is the so-called spectral action introduced by 
Chamseddine and Connes \cite{CC,CC1,CC2}
\begin{align}
\label{spectralaction}
S(\DD,\Lambda,f)=\Tr\big( f( \DD^2 /\Lambda^2)\big)
\end{align} 
and defined for $\Lambda \in \R^+$ which plays the role of a cut-off (and needed to make $\DD/\Lambda$ 
dimensionless) and for a function $f$ such that, of course, 
$ f( \DD^2 /\Lambda^2)$ is trace-class. In general, one chooses $f\geq 0$ since the action 
$\Tr\big( f( \DD^2 /\Lambda^2)\big)\geq 0$ will have the correct sign for an Euclidean action. This action is the 
appropriate one in the framework of noncommutative geometry to reproduce several physical situations like the 
Einstein-Hilbert action in gravitation or the Yang-Mills-Higgs action in the standard model of particle physics 
\cite{ ConnesMarcolli}, and the positivity of the function $f$ implies positivity of actions for gravity, Yang-Mills or 
Higgs couplings, and the Higgs mass term is negative.

Till the end of this Section we shall present a non-technical summary of our results to give a more physics-oriented 
reader a chance to appreciate them without going through the mathematics of the rest of this paper.
 
Let $M=\mathbb{R}^{2m}$ be an even dimensional real plane, $d=2m\geq 2$, endowed with a spin structure given 
by the spinor bundle $S=\C^{2^m}$. We denote by $\DD$ the free Dirac operator and by  $\DD_A$ the standard 
Dirac operator with a gauge connection $A$ acting on the Hilbert space $\H\vc L^2(M,S)$. 
We will use the notations and conventions from \cite[eq. (3.26)]{Vassilevich:2003xt}, namely in local coordinates
\begin{align}
\label{Dirac}
\DD\vc i\gamma^\mu \partial_\mu,
\qquad  \DD_A \vc i\gamma^\mu \nabla_\mu \vc i\gamma^\mu (\partial_\mu + A_\mu)
\end{align}
where $A_\mu \in \Gamma\big(M,\text{End}(S)\big)$ is taken in some representation of complex dimension $N$ of the gauge group. 
We assume that $A_\mu$ with field strength 
$F_{\mu\nu}\vc \partial_\mu A_\nu - \partial_\nu A_\mu +[A_\mu,A_\nu ]$ satisfies for a given 
$\rho>d$,
\begin{align}
\label{hyp}
\vert A_{\mu}(x) \vert \leq c\,(1+\vert x\vert )^{-\rho}\,, \qquad \del^\b A_\mu(x)=\mathcal{O}(1)\, ,
\end{align}
for all $\b\in \N^d$ with $|\b|\leq m$. This hypothesis will be justified in Section \ref{Spectral shift function}.
Since the manifold $M=\mathbb{R}^{2m}$ is non-compact, $f(\DD^2)$ is never trace-class, so, in order to 
get rid of trivial volume divergence, we modify \eqref{spectralaction} as
\begin{equation}
\label{spectral action}
S(\DD_A,\Lambda,f)\vc\Tr\big( f( \DD_A^2 /\Lambda^2)-f(\DD^2/\Lambda^2) \big).
\end{equation}
It is known that $(\SS(\R^d),\H,\DD)$ is a spectral triple with the non-unital algebra of Schwartz functions, 
see definition in \cite{CGRS,GGBISV} and \cite{GGBISV} for a proof. Other algebras are possible, like smooth 
functions on $\R^d$ with arbitrary partial derivatives bounded and integrable. 
Moreover, \eqref{spectral action} can be seen as a variation of the spectral actions when a one-form is turned on, 
which also makes sense for non-unital spectral triples. Actually, the two spectral triples 
$(\A,\H,\DD)$ and $(\A,\H,\DD_A)$ represent the same geometry, because obtained via Morita equivalence 
\cite{ConnesMarcolli,Polaris}. In the heat kernel asymptotics when $\Lambda \to \infty$, the coefficient of 
$\Lambda^0$ will be given by $\zeta_{\DD_{A}}(0)-\zeta_\DD(0)$ which has been 
computed in \cite{CC1} (in case $\A$ has a unit) and the kernels of $\DD_A$ and $\DD$ will not appear. 
Unfortunately, this does not give access to a definition of spectral action
for a generic non-compact spectral triple, but it avoids the use of an extra smearing 
function $g$ which can regularize $\Tr \big(g\,f(\DD^2/\Lambda^2) \big)$. Such $g$ has been used for instance in 
\cite{GilkeyNew, Vassilevich:2003xt,GI}. 

The trace formula \eqref{trace formula} gives  
$S(\DD_A,\Lambda,f)=\int_{-\infty}^\infty \xi(\lambda)\,f'(\lambda)\,d\lambda$ for the spectral shift function 
$\xi(\lambda)$. Thus \eqref{spectral action} makes sense for a large class $\CC_1$ of functions $f$ 
described in Section \ref{Spectral shift function}.

The spectral action is quite well studied in the framework of the \emph{weak-field expansion}, which is 
constructed in the following way (see \cite{ConnesMarcolli} for a detailed treatment).
It is assumed now that $f$ is a Laplace transform of a function $\varphi$:
\begin{equation}
f(z)=L\varphi(z)\vc\int_0^\infty dt\, e^{-tz} \varphi (t)\,.\label{Laplace}
\end{equation}
Then,
\begin{equation}
S(\DD_A,\Lambda,f)=\int_0^\infty dt\,\tilde K(t/\Lambda^2,\DD_A^2)\, \varphi (t)\,,
\label{SHK}
\end{equation}
where
\begin{equation}
\tilde K(s,\DD_A^2)=\Tr \Big( e^{-s\DD_A^2}-e^{-s\DD^2} \Big)
\label{subHK}
\end{equation}
is the subtracted heat kernel, which has an asymptotic expansion as $s\to +0$
\begin{equation}
\tilde K(s,\DD_A^2) \, \underset{s\downarrow 0}{\sim} \,\sum_{k=1}^\infty s^{-m+k}\,a_{2k}(\DD_A^2).
\label{asymp}
\end{equation}
Then
\begin{equation}
\label{asS}
S(\DD_A,\Lambda,f) \, \underset{s\downarrow 0}{\sim} \, \sum_{k=1}^\infty \Lambda^{2(m-k)} \,
\varphi_{2k} \, a_{2k}(\DD_A^2)
\end{equation}
with
\begin{equation}
\varphi_{2k}=\int_0^\infty dt\, t^{-m+k}\varphi (t) \,.\label{phi2k}
\end{equation}
The heat kernel coefficients $a_n$ are very well known at least in the commutative
case \cite{GilkeyNew,Vassilevich:2003xt}. Let us assign canonical mass dimension $1$
to $A_\mu$ and to the derivative. Each coefficient $a_{2k}$ is an integral of a polynomial
of the connection $A_\mu$ and its' derivatives of the total canonical dimension $2k$.
The expansion (\ref{asS}) is valid therefore when the fields and their derivatives
are small compared to $\Lambda$. Hence, this is a weak-field expansion.

The expansions (\ref{asymp}) and (\ref{asS}) do not contain the term with $a_0$,
which cancels out between the $\DD_A$ and $\DD$. This is, however, not a generic
feature of spectral actions, but is rather a consequence of taking the same
flat Euclidean metric on $\mathbb{R}^d$ in both operators. In general, one has to
allow for metric perturbations, see e.g. \cite{ConnesMarcolli}. Let us consider
the metric $g_{\mu\nu}=\delta_{\mu\nu}+h_{\mu\nu}$ where the fluctuations 
$h_{\mu\nu}$ over the unit metric are well localized (fall off sufficiently fast
at infinity). Let us introduce this metric in $\DD_A$ through suitably defined
$\gamma$-matrices and add corresponding connection terms. Let us leave $\DD$ as above.
At this point we have to assume that the trace in (\ref{spectral action}) exists,
though the case of metric fluctuations is not covered by the analysis of Section
\ref{Spectral shift function}.
Then the summation in (\ref{asymp}) and (\ref{asS}) have to be extended to $k=0$
with $a_0\sim \int dx (\sqrt{\det{g}} -1)$. Variation of $h_{\mu\nu}$ in the $a_0$ term in the
spectral action produces the standard cosmological term in the equations of
motion. Note, that the Einstein action with a cosmological term is always infinite
on its' classical solutions on non-compact manifolds. Therefore, the necessity
to subtract a contribution from some reference metric is well understood in
General Relativity. For us here, it is only important that the $a_0$ term reappears
in the spectral action with metric fluctuations and gives rise to the cosmological
constant. This fact will be used for a physical interpretation of our results, see
Remark \ref{dual}.

If the spectral action is to be taken seriously, one should also study
it beyond the weak-field approximation. In particular, the terms which depend quadratically
on the field strength and contain arbitrary number of derivatives govern the ultraviolet
behavior of propagators and are utterly important for quantization. 
One observes that, if one restricts himself to the second order in $F_{\mu\nu}$, the heat kernel
coefficients have the form $a_{2k+4}\sim F_{\mu\nu}\Delta^k F_{\mu\nu}$
where $\Delta=-\partial_\mu\partial^\mu$ is the free Laplacian.
After a Fourier transform, this can be translated to $a_{2k+4}\sim \hat F_{\mu\nu}(-p)\, p^{2k} \hat F_{\mu\nu}(p)$.
Then, it was suggested in \cite{Suijlekom2} to pick up a function $f$
such that the coefficients $\varphi_{2k}$ vanish for $2k>2m+N$, while $\varphi_{2m+N}\ne 0$
with a sufficiently large $N$. The spectral action restricted to first few coefficients in the heat kernel
expansion displays the polynomial growth $\sim p^N$ and the corresponding propagators decay at large momenta.
Power counting arguments show that the restricted Yang-Mills theory becomes super-renormalizable 
\cite{Suijlekom2}. 

Using the Barvinsky-Vilkovisky approach of covariant perturbation theory 
\cite{Barvinsky:1987uw,Barvinsky:1990up}, the aim of this 
work is to compute the action \eqref{spectral action} as a function of the field strength $F$ of the connection 
$A$ in full generality, then to control its ultraviolet regime. We restrict ourselves to the second order in $F$. 
To this order, we derive a remarkably simple formula for the spectral spectral action for a large class of function  $f$. 
Then, we show that the spectral action decays as $1/p^4$ in the ultraviolet asymptotics. Therefore, the 
propagator of the Yang-Mills field grows at large momenta, and the full Yang-Mills spectral action
is not super-renormalizable in contrast to its expansion considered in \cite{Suijlekom2}.

This paper is organized as follows. 
Section \ref{Spectral shift function} recalls few facts on the regularization deduced from Krein spectral shift function 
and conditions on $f$ which guarantee that $f( \DD_A^2 /\Lambda^2)-f(\DD^2/\Lambda^2)$ is a trace-class 
operator. Few known results where $f$ is piecewise continuous are quoted. We use the $G$-pseudodifferential 
calculus to study some sufficient condition on the connection $A$ that gives a trace-class resolvent perturbation. 

In Section \ref{sec-spact2}, we use the covariant perturbation expansion to compute the spectral action to the 
second order in $F(A)$ and its ultraviolet asymptotics. Our main result for the spectral action at the second
order in $F$ reads 
\begin{equation*}
S(\DD_A,\Lambda,f)(F) =
 \tfrac{\Lambda^{d-4}}{(4\pi)^m}\int_M  d^{2m}p\, {\rm tr} \bigl[ \hat F^{\mu\nu}(-p) \,w_\Lambda(p^2)\, \hat F_{\mu\nu}(p) \bigr] 
+ \mathcal{O}(F^3)\,
\end{equation*}
with the form-factor $w_\Lambda(p^2)$ given in Lemma \ref{lemmawL} as
\begin{equation*}
w_\Lambda(p^2)=(-1)^{m-1}\,2^{m-2}\int_0^1d\a\,\Big[f^{[m-2]} \big(\a(1-\a)\Lambda^{-2}p^2\big)
-\tfrac{2\Lambda^{2}}{p^2} \int^{\a(1-\a)\Lambda^{-2}p^2}_0 ds_1\, f^{[m-2]}(s_1) \Big]. 
\end{equation*}
Here $f^{[n]}$ is the $n$th primitive of the function $f$. Dependence on the dimension $d=2m$ of
underlying space resides in the overall numerical factor and in the number of repeated integrations of
$f$. The ultraviolet asymptotics
(Theorem \ref{Result}) of $w_\Lambda(p^2)$ depends on $m$ through a numerical factor only.

Section \ref{sec-ass} discusses how to relax hypothesis on the function $f$. In particular, if $f$ is not assumed 
to be positive, one can construct a spectral action which decays faster as $1/p^4$. We also treat a 
step-function cut off. 

Higher terms in the $F$-expansion are briefly considered in section \ref{sec-ht}. 

Concluding remarks are given in section \ref{sec-con}.

\section{Non-compact spectral action and spectral shift function}
\label{Spectral shift function}
The main purpose of this Section is to show that the trace in (\ref{spectral action})
exists under certain assumptions on $f$ and $A$. We have start with a short tour into
the theory of Krein spectral shift function.

\subsection{Spectral shift function}

The Krein spectral shift functions are widely used in scattering theory, see 
\cite{Pushnitski09,Pushnitski10,Simon,Yafaev92,Yafaev05,Yafaev07}.  
For a pair of selfadjoint operators $H_0,H$ on a Hilbert space $\H$, the spectral shift function  
$\xi(\lambda)$ is defined by the Lifshits trace formula 
\begin{align}
\label{trace formula}
\Tr\big(f(H)-f(H_0)\big)=\int_{-\infty}^\infty d\lambda\,\, \xi(\lambda)\,f'(\lambda).
\end{align}
The idea behind is that, if $f(\lambda)=\chi_{]-\infty,\lambda[}$, then $\xi(\lambda)=
-``\Tr "\big( P(]-\infty,\lambda[)-P_0(]-\infty,\lambda[)\big)$ 
where $P, P_0$ are the spectral projections of $H,H_0$ and $``\Tr"$ is some regularized trace, so  
$\xi(\lambda)$ appears to be a regularization of the difference of eigenvalue counting functions. 
Remark that $\xi(\lambda-0)$ is equal to the index of the Fredholm 
pair $P(]-\infty,\lambda[)$ and $P_0(]-\infty,\lambda[)$ and coincides with the spectral flow of $H$ and $H_0$ 
when $\lambda$ is in the discrete spectrum of $H_0$. 
These are different regularizations of $\Tr \big( P(]-\infty,\lambda[)-P_0(]-\infty,\lambda[)\big)$: the $\xi(\lambda)$ 
function is the regularization obtained by replacing the difference of spectral projections by 
$f(H)-f(H_0)$, where $f$ is a smooth approximation of $\chi_{]-\infty,\lambda[}$, while the index is obtained by 
replacing $\Tr$ by index. These two regularizations does not coincide when $\lambda$ is in 
the essential spectrum of $H_0$ since the function $\xi$ is not always integer-valued as the index.

Assume that $V \vc H-H_0$ is in $\L^1(\H)$ (trace-class operators). If $R(X)(z)$ is the resolvent of the operator $X$, 
the function $D(z)\vc \Det \big(1+V\,R(H_0)(z)\big)$ is holomorphic and,
moreover, satisfies $D^{-1}(z)D'(z)=\Tr\big(R(H_0)(z)-R(H)(z)\big)$. Defining 
\begin{align*}
\xi(\lambda;H,H_0)\vc \pi^{-1}\,\underset{\epsilon \downarrow 0}{\lim} \, \arg\, D(\lambda+i\epsilon)
\end{align*} 
for almost all 
$\lambda$, we get the following: 

\centerline{$\log\,D(z)=\int_{-\infty}^\infty d\lambda\,\,\xi(\lambda)\,(\lambda-z)^{-1}$ for 
$\Im(z)\neq 0$, $\int_{-\infty}^\infty d\lambda\, \vert \xi(\lambda)\vert\leq \norm{V}_1$ and 
$\Tr(V)=\int_{-\infty}^\infty d\lambda\,\,\xi(\lambda)$.} Moreover, the trace formula holds at least for functions 
$f \in C^\infty_c(\R)$ which are smooth with compact support.

Clearly, $\xi$ deserves its name of spectral shift since when $\lambda$ is an isolated eigenvalue for $H$ 
and $H_0$, with multiplicity $n$ and $n_0$, $\xi(\lambda+0)-\xi(\lambda-0)=n_0-n$ and $\xi$ gets 
constant integer values in any interval located in the resolvent sets of both $H$ and $H_0$.

Before applying this to our situation, we assume the existence of $c\in \R$ such that $H_0+c1$ and $H+c1$ 
are positive definite, and 
\begin{align}
\label{Resolvent}
\big(R(H)\big)^n(z)-\big(R(H_0)\big)^n(z)\in\L^1(\H)
\end{align}
for some $z\le -c$ and $n\in\N$. 

Now, define $\xi$ by 
\begin{align}
\label{xi}
\xi(\lambda) \vc -\xi \big((\lambda+c)^{-n};(H+c1)^{-n},(H_0+c1)^{-n}\big)
\end{align} 
for $\lambda>-c$ and $\xi(\lambda)=0$ for $\lambda \leq -c$. \\
Let $\CC_{1,n}$  be the set of functions $f$ having two locally bounded derivatives and satisfying 
\begin{align*}
\big(\lambda^{n+1}\,f'(\lambda)\big)'(\lambda)\, \underset{\lambda \to \infty}{\sim}\,\OO(\lambda^{-1-\epsilon})  
\quad \text{for some }\epsilon >0.
\end{align*}
When $f\in \CC_{1,n}$,  $f(H)-f(H_0)\in \L^1(\H)$ and the trace formula \eqref{trace formula} holds true.

\begin{remark}
For instance, $f$ defined by $f(\lambda)=(\lambda-z)^{-n}$ for $z$ in the resolvent set of $H$ and $H_0$ is in 
$\CC_{1,n}$; typically, $f_r(\lambda)=(\lambda+a)^{-r}$ is in $\CC_{1,n}$ for $r>n$ and $a>0$.

Moreover,  
$\Tr\big(R^n(z)-R_0^n(z)\big)=-n\int_{-\infty}^\infty d\lambda \,\xi(\lambda)\,(\lambda-z)^{-(n+1)}$. 
Similarly, for $b<\inf \text{Spect}(H)$,
\begin{align*}
\Tr\big(e^{-t\,H}-e^{-t\,H_0}\big)&=-t\int_{-\infty}^\infty\ d\lambda\,\,\xi(\lambda)\,e^{-t\,\lambda} \\
&=(2\pi i)^{-1}(n-1)!\;t^{-n+1}\int_{b-i\infty}^{b+i\infty} dz\,\Tr\big(R^n(z)-R_0^n(z)\big) \,e^{-t\,z}.
\end{align*}
\end{remark}

For any triplet of selfadjoint operators $H_0, H_1,H_2$ with trace-class differences, the equality 
$\xi(\lambda;H_2,H_0)= \xi(\lambda;H_2,H_1)+\xi(\lambda;H_1,H_0)$ is equivalent to the additivity of 
spectral action \eqref{spectral action}.

We shall now apply the above results to $H=\DD_A^2$ and $H_0=\DD^2$. It is known that $H_0$ has a purely 
absolutely continuous spectrum $\R^+$ with infinite multiplicity for $d \geq 2$.

\subsection{\texorpdfstring{Trace-class resolvent perturbation and $G$-pseudodifferential calculus}{Trace-class 
resolvent perturbation and G-pseudodifferential calculus}}

The goal of this section is to prove the following
\label{Gsec}
\begin{theorem}
\label{theoremA1}
Let $H=\DD_A^2$ and $H_0=\DD^2$ defined in \eqref{Dirac}, with $A$ satisfying \eqref{hyp}.

If $c>0$,
 then condition \eqref{Resolvent} 
is satisfied with $n=m(=d/2)$, so that $\xi$ as defined in \eqref{xi} exists and Eq.\ \eqref{trace formula} holds true for 
$f\in \CC_1:=\CC_{1,m}$. In particular, the variation of the non-compact spectral action \eqref{spectral action} is 
well-defined for any $f\in \CC_{1}$ and $\Lambda\in \R^+$.
\end{theorem}

We start with a lemma about resolvent perturbations for abstract selfadjoint operators on a Hilbert space. 
We denote $R_T:c\in \R^+ \to (T^2+c)^{-1}\in \B(\H)$ for any selfadjoint operator $T$ on a Hilbert space $\H$.

\begin{lemma}
\label{lemresolvent}
Let $\H$ be a Hilbert space, $P$ be an unbounded selfadjoint operator on $\H$ and 
denote for any $k\in \N$, $H_P^k$ the Hilbert space based on $\Dom P^k$, endowed with the 
scalar product 
$$
(\psi,\phi)_k:= ((P-i)^k\psi, (P-i)^k \phi)_\H +(\psi,\phi)_\H. 
$$
Let $n\in \N$ and $A$ be a selfadjoint bounded operator on $\H$ such that $A$ sends continuously $H_P^k$ 
into itself for any $k\in \set{0,\cdots n}$.
\\
The following holds for any $k\in \set{0,\cdots,n}$:

(i) $H_{P+A}^k = H_P^k$ (with equivalent norms).

(ii) For any $c>0$, $(R_P(c))^{-k} (R_{P+A}(c))^k $ is a bounded operator on $\H$. Moreover,
$$
\big(R_{P}(c))^k-(R_{P+A}(c)\big)^k = \sum_{j=1}^{k} \big(R_P(c)\big)^{j} (AP+PA+A^2) 
\big(R_{P+A}(c)\big)^{k+1-j}\, .
$$

(iii) The operator $\big(R_{P}(c)\big)^k-\big(R_{P+A}(c)\big)^k$ is trace-class on $\H$ if for any $j\in \set{1,\cdots,k}$, 
the operator $\big(R_P(c)\big)^{j} (AP+PA+A^2) \big(R_{P}(c)\big)^{k+1-j}$ is trace-class on $\H$. 

\end{lemma}
\begin{proof}
$(i)$  Let us check by induction that 
$\Dom P^k= \Dom (P+A)^k$. It is clearly true for $k=0,1$. Suppose that it holds for a given $1\leq k\leq n-1$. 
We see that if $\psi\in \Dom (P+A)^{k+1}$ then $\psi\in \Dom P^k$ is such that $(P+A)\psi \in \Dom P$. But since 
$A\Dom P^k \subseteq \Dom P^k \subseteq \Dom P$, we get $P\psi\in \Dom P$, which implies 
$\psi\in \Dom P^{k+1}$. Similarly, we get the other inclusion $\Dom P^{k+1}\subseteq \Dom (P+A)^{k+1}$. 

We now look at the norm equivalence. Fix $k$ with $0\leq k \leq n$. We first claim that the operator 
$B:=(P+A-i)^k(P-i)^{-k}$ is bounded on $\H$. 
To prove this, we take $\psi\in \Dom P^k = \Dom (P+A)^k$ and we expand $(P+A-i)^k\psi = (P-i)^k\psi +X\psi$ 
where $X$ is a sum of terms of the form
$\prod_{p=1}^k X_p$ where each $X_p$ is either $A$ or $P-i$, and there is a least one $p$ such that $X_p =A$. 
Since $A$ sends continuously $H_P^m$ into itself for any $m\in \set{0,\cdots,n}$, we see that $X$ sends 
continuously $H_P^{k}$ into $H_P^1$. In particular, $X(P-i)^{-k}$ 
is bounded on $\H$. If $\psi\in \H$, $(P-i)^{-k}\psi \in \Dom P^k$, and thus $B\psi= \psi + X(P-i)^{-k}\psi$, which 
proves the claim. The bounded inverse theorem now entails that $B^{-1}=(P-i)^{k}(P+A-i)^{-k}$ is also bounded 
on $\H$.

Denoting $\norm{\cdot}_{A,k}$ (resp. $\norm{\cdot}_k$ the Hilbert norm associated to $H_{P+A}^k$ (resp. $H_P^k$), 
we see that for any $\psi\in \Dom P^k=\Dom (P+A)^k$, 
\begin{align*}
&\norm{\psi}^2_{A,k}\leq \max\set{1,\norm{B}^2} \norm{\psi}^2_k \,, \qquad\norm{\psi}^2_k \leq \max \set{1,
\norm{B^{-1}}^2}\norm{\psi}^2_{A,k}
\end{align*}
from which the result follows.

$(ii)$ Clearly $\big((P+A)^2+c\big)^{-k}$ sends continuously $\H$ into $H_{P+A}^{2k}$. Therefore, by $(i)$, it is 
continuous from $\H$ into $H_P^{2k}$. 
Moreover, $P^{2k}$ sends continuously $H_P^{2k}$ into $\H$. As a consequence, the composition 
$\big(R_P(c)\big)^{-k} \big(R_{P+A}(c)\big)^k $ is a bounded operator on $\H$. 

For the second statement, we use ${R}'(c)=\big(R(c)\big)^2$ after $k$ derivations of the resolvent identity 
$R_P(c)-R_{P+A}(c)=R_P(c)(AP+PA+A^2)R_{P+A}(c)$ with respect to the parameter $c$.

$(iii)$ follows directly from $(ii)$.
\end{proof}

We now recall some definitions and properties of $G$-pseudodifferential operators 
on $\R^d$ \cite{NR}: let $G^{p,q}(\R^d)$ (resp. $OPG^{p,q}$) be the $G$-class of symbols (resp. pseudodifferential 
operators) of order $(p,q)$, valued in $\M_{2^m}(\C)$, with $d=2m$; recall that $a\in G^{p,q}(\R^d)$ if and only if 
for any multi-index $(\a,\b)\in \N^{2d}$,
$$
\norm{\del_\xi^\a\del_x^\b a (x,\xi)}\lesssim 
\langle \xi\rangle^{p-|\a|}\langle x\rangle^{q-|\b|}
$$
where $\langle y \rangle \vc (1+\norm{y}^2)^{1/2}$.

We denote $\RR_\rho$, where $\rho>0$, the space of all bounded operators operators on 
$\H:= L^2(\R^d,\C^{N2^m})$ that 

- sends continuously $H^k:=H^k(\R^d,\C^{N2^m})$ into itself for all $k\in \set{0,\cdots,m}$. 

- extends as bounded operator from $L^{2,-s}$ into $L^{2,-s+\rho}$ for any $s>0$, where we define the set 
$L^{2,s}:=L^{2,s}(\R^d,\C^{N2^m})$ as the weighted Hilbert space of functions such that 
$\norm{\langle x\rangle^{s}f}_{\H}<\infty$.

We recall that $N$ denotes dimension of the gauge group representation.

The interest of the space $\RR_\rho$ lies in the following result.
\begin{lemma}
\label{lemGcalcul}
Let $\wt P$ be an elliptic symmetric operator in $OPG^{1,0}$ and $A\in \RR_\rho$ be selfadjoint with $\rho>d$. 
Denote $P$ the selfadjoint unbounded operator acting as $\wt P$ on $\H$ with domain $H^1$.\\
Then for any $c>0$,
$$
\big(R_P(c)\big)^m - \big(R_{P+A}(c)\big)^m \in \L^1(\H)\, .
$$
\end{lemma}

\begin{proof}
Note that we can suppose without loss of generality that $\rho\leq d+1$. 

Since $\wt P\in OPG^{1,0}$ is elliptic, $H_P^k=H^k$ for any $k\in \N$. By Lemma \ref{lemresolvent} $(iii)$, 
since it is supposed that $A$ sends $H^k$ into itself for any $0\leq k\leq m$, it is sufficient to check that for 
any $j\in \set{1,\cdots,m}$, the operator $\big(R_P(c)\big)^{j} (AP+PA+A^2) \big(R_{P}(c)\big)^{m+1-j}$ is 
trace-class on $\H$. 

Fix $1\leq j\leq m$. We will now prove that $\big(R_P(c)\big)^{j} AP\big(R_{P}(c)\big)^{m+1-j} \in \L^1(\H)$. 
The proof that $\big(R_P(c)\big)^{j}PA\big(R_{P}(c)\big)^{m+1-j}$ and $\big(R_P(c)\big)^{j} A^2 (R_{P}(c))^{m+1-j}$ 
are in $\L^1(\H)$ is similar.   
 
Let $M_s$ be the pseudodifferential operator with Weyl symbol $\langle x\rangle^s$, for $s\in \R$. 
The operator $M_s$ is elliptic invertible in $OPG^{0,s}$, with inverse $M_{-s}$, and sends 
continuously $L^{2,s'}$ into $L^{2,s'-s}$, for any $s'\in \R$. 
By hypothesis $A$ can be extended so that $M_{2j} A M_{\rho-2j}$ is bounded on $\H$. In other words, 
there exists a bounded operator $B\in \B(\H)$ such that $A=M_{-2j} B M_{-\rho+2j}$. 

Note also that $\big(R_P(c)\big)^{j}$ extends as pseudodifferential operator in $OPG^{-2j,0}$ since $\wt P$ is 
elliptic in $OPG^{1,0}$.

Fix $0<\eps<\min\set{\half(\rho/d-1), 2j/d, (\rho-2j)/d}$. It is known \cite{BuzanoNicola, BuzanoToft} that 
if $T\in OPG^{p,q}$, $p<-d/r$, $q<-d/r$, then $T \in \L^r(\H)$. Thus, it appears that 
$\big(R_P(c)\big)^j M_{-2j} \in \L^{p_1}(\H)$ where $p_1^{-1}=2j/d-\eps$ and 
$M_{-\rho+2j} P\big(R_{P}(c)\big)^{m+1-j}\in \L^{p_2}(\H)$ where $p_2^{-1}=(\rho-2j)/d-\eps$. As a consequence, 
\begin{align*}
\big(R_P(c)\big)^{j} AP\big(R_{P}(c)\big)^{m+1-j} = \big(R_P(c)\big)^{j} M_{-2j} \, B\, 
M_{-\rho+2j} P\big(R_{P}(c)\big)^{m+1-j} \in \L^{(p_1^{-1}+p_2^{-1})^{-1}}(\H)
\end{align*}
and the result now follows from the fact that $p_1^{-1}+p_2^{-1} =\rho/d-2\eps  \geq 1$. 
\end{proof}

\begin{proof}[Proof of Theorem \ref{theoremA1}] The condition \eqref{hyp} implies that the multiplication operator 
$A$ belongs to $\RR_\rho$. Thus, Lemma \ref{lemGcalcul} applied to $\wt P=\DD$ yields directly the result.
\end{proof}

\begin{remark} When $A$ is a multiplication operator, the condition $A\in \RR_\rho$ implies some constraints on 
the derivatives of $A$. To remove these constraints, a technique based on a commutation between 
$R_{P+A}$ and $M_s$ was used in \cite{Yafaev05}. 
On the other hand, our class $\RR_\rho$ contains operators that are not necessarily multiplication by functions. 
For instance $M_{-\rho}P \in \RR_\rho$ when $P$ is a pseudodifferential operator on $\R^d$ with symbol in the 
H\"{o}rmander class $S_{0,0}^0$.
\end{remark}

\begin{remark} 
In the stronger case where $A$ is supposed to be a pseudodifferential operator in $OPG^{0,-\rho}$ where 
$\rho>d$, then the trace-class property $\big(R_P(c)\big)^m - \big(R_{P+A}(c)\big)^m \in \L^1(\H)$ follows 
more directly and there is a integral formula for the value of trace of the resolvent perturbation 
$\big(R_P(c)\big)^m - \big(R_{P+A}(c)\big)^m$.

Indeed, in this case, $P+A$ is elliptic in $OPG^{1,0}$. By ellipticity and spectral invariance, $R_A$ and $R$ 
are in $OPG^{-2,0}$. Since $\rho>0$, $AP+PA+A^2\in OPG^{1,-\rho}$ implies
that $R_P^m-R_{P+A}^m\in OPG^{-2n-1,-\rho}$. The result now follows from the following property 
(see for instance \cite[Theorem 4.4.21]{NR} for more details and greater generality):
if $T\in OPG^{p,q}$ with $p<-d$, $q<-d$, then $T$ is a trace-class operator on $L^2(\R^d, \C^{N2^m})$, and 
moreover $\Tr(T) = \int_{\R^{2d}} \tr_{\C^{N2^m}}(T_w)$ where $T_w$ is the Weyl symbol of $T$ which belongs
to $G^{-d,-d}(\R^d)\subseteq L^1(\R^{2d}, \C^{N2^m})$. 
\end{remark}

\section{Spectral action to the second order of field strength}\label{sec-spact2}

\subsection{Heat kernel to the second order of field strength}\label{sec-BV}

Despite the fact that in \cite{Barvinsky:1990up}, $M$ is only supposed to be asymptotically flat, we assume here 
that $M$ is flat. Equation \eqref{Dirac} yields
\begin{gather*}
\DD_A^2=-(\nabla^2+E),\\
\nabla^2 \vc g^{\mu \nu} \, \nabla_\mu \nabla_\nu\,,
\qquad E \vc  \tfrac{1}{4}\,[\ga^\mu,\ga^\nu]\,F_{\mu\nu} \,.
\end{gather*}
We first compute the trace of the heat kernel as a function of $F$. In \cite{Barvinsky:1990up}, Barvinsky 
and Vilkovisky have introduced a covariant perturbation theory, and the use of  \cite[eq. (2.1)]{Barvinsky:1990up} 
gives, for any $s>0$, small or large,
\begin{align}
\tilde K(s,\DD_A^2)(F)&=\tfrac{1}{(4\pi s )^m}  \int_M  \,d^{2m}x \,
\tr\, \Bigl[sE+ s^2\,E\,\tfrac{1}{2}h(-s\partial^2)\,E 
+ s^2\,F_{\mu\nu}\, q(-s\partial^2) \,F^{\mu\nu} \Bigr]   +\mathcal{O}(F^3)(s) \label{traceheat}
\end{align}
where the trace is both on the gauge and spinor indices and
\begin{align}
&q(z)\vc -\tfrac{1}{2} \tfrac{h(z)-1}{z} \,,\nonumber\\
&h(z)\vc\int_0^1 d\alpha \,e^{-\alpha (1-\alpha)\,z}\,. \label{h}
\end{align}
The function $h$ appears here since it corresponds to an expansion process: 
\begin{align*}
h(z)=\int_0^\infty\int_0^\infty d\a_1\,d\a_2\,\delta(1-\a_1-\a_2)\,e^{-\a_1\a_1z}.
\end{align*}
The main interest of formula \eqref{traceheat} is that it is valid for all positive values of $s$, and not only when 
$s\to 0$. 
Thus if
\begin{align}
& \kappa(z) \vc 2^{m-2}\big(-h(z) +4q(z)\big)\,, \label{kappadef}
\end{align}
we get, with 
$\tr\big( [\ga^\mu,\ga^\nu]\,[\ga^\rho,\ga^\sigma] \big)=2^{m+2}(g^{\mu\sigma}g^{\nu\rho}-g^{\mu\rho}g^{\nu\sigma}),$
\begin{align}
\label{Trace}
\Tr \big( \exp(-s\DD_A^2)-\exp(-s\DD^2)\big)(F)&=\tfrac{1}{(4\pi )^m}\,s^{-m+2}\,\int_M d^{2m}x \,
\tr \Bigl[ F_{\mu\nu} \,\kappa(-s\partial^2)\, F^{\mu\nu} \Bigr] +\mathcal{O}(F^3)(s)
\end{align}
where the trace is now over gauge indices only.
Note that 
\begin{align}
h(z)=2\sum_{n=0}^\infty \tfrac{z^n}{n!}\,\big( \tfrac{d\,}{dz}\big)^{2n}\,\int_0^{1/2} d\a \,e^{-\a z}\, 
\underset{z\to \infty}{\sim}\, 2\sum_{j=0} \tfrac{z^{-1-j} (2j)!}{j!}\, .\label{hzas1}
\end{align}
So we obtain the asymptotics of $\kappa(z)=2^m\, h'(z)$:
\begin{align}
\label{kappa}
\kappa(z)  \underset{z \to \infty}{\sim}\, -2^{m+1}\, z^{-2}\,.
\end{align}

\subsection{\texorpdfstring{Assumptions on the function $f$}{Assumptions on the function f}}
\label{sec-ncom}

Let $\CC_2$ be the class of functions $f$ such that
\begin{align}
\label{hypothesis}
&f: [0,\infty[ \to \R^+ \text{ is piecewise continuous and } \nonumber\\ 
&f(x)\, \underset{x\to\infty}{\sim}\,\mathcal{O} \big(f_\epsilon(x)\big) \text{ with } 
f_\epsilon(x) \vc (x+\epsilon)^{-m-\epsilon} \text{ for some } \epsilon >0
\end{align}
where we take care of the eventually (finite dimensional) kernel of $\DD^2$ via $x+\epsilon$.

Let $\CC_3$ be the class of functions $f$ such that 
\begin{align}
\label{LT}
&f \text{ is a function (not a distribution) such that exists a function $\varphi$ satisfying}\nonumber\\ 
&f=L\varphi \text{ and } \nonumber \\
&\qquad \text{either } \varphi(t)  \text{ is a function and }t^{-k}\,\varphi(t) \in \L^1([0,\infty[) 
\text{ for } k\in \set{1,\cdots,m} \nonumber\\
& \qquad \text{or } \varphi \text{ is a finite sum of } t^p\delta^{(n)}(t-a) \text{ for arbitrary } n,p\in\N,\,a\in ]0,\infty[.
\end{align}
The set $\CC_1\cap \CC_2$ contains for instance $(a+x)^n\,e^{-bx}$ for $a,\,b\geq 0,\,n\in \N$, or completely 
monotonic functions (see Definition \ref{defcm}) which fall off like in \eqref{hypothesis}. We can enlarge this 
set with all finite positive linear combinations of functions satisfying (\ref{LT}).

\subsection{Calculation of the spectral action}

We first assume that $f=L\varphi$ is a Laplace transform of a function $\varphi$, as in (\ref{Laplace}), so 
the spectral action is given by (\ref{SHK}).

The expansion (\ref{asS}) starts with $k=1$. However, the same universal function $f$ should describe the 
spectral action in the presence of metric fluctuations when the heat kernel coefficient $a_0$ does not vanish.
For this reason, $\varphi_0$ must be finite,
\begin{align}
\label{cf}
c_f \vc \varphi_0= \int_0^\infty dt\, t^{-m}\, \varphi(f)(t)<\infty.
\end{align}
With the definition
\begin{align}
\label{wLambda}
w_\Lambda(z) \vc \int_0^\infty  dt\,t^{-m+2}\,\kappa(t\Lambda^{-2}\,z)\,\varphi(t),
\end{align}
we obtain 
\begin{equation}
\int_M d^{2m}x\, {\rm tr}\bigl[ F^{\mu\nu} \,w_\Lambda (-\partial^2) \, F_{\mu\nu} \bigr] 
=\int  d^{2m}p\, {\rm tr} \bigl[ \hat F^{\mu\nu}(-p) \,w_\Lambda(p^2)\, \hat F_{\mu\nu}(p) \bigr]
\label{wp}
\end{equation}
where $\hat F$ is the Fourier transform of $F$.
At this point, it is important to quote that, by the analysis in \cite{Barvinsky:1990up}, at all order in the curvature, 
the large $s$ behavior of $\Tr \big( \exp(-s\DD_A^2)\big)(F)$ in \eqref{Trace} is in $s^{-m+2}$. In particular, in 
dimension 4, it stabilizes.

So the $F$-dependence of the spectral action (assuming commutation of integrals) is
\begin{align}
\label{Squadra}
S(\DD_A,\Lambda,f)(F) &=
 \tfrac{\Lambda^{d-4}}{(4\pi)^m}\int_M  d^{2m}p\, {\rm tr} \bigl[ \hat F^{\mu\nu}(-p) \,w_\Lambda(p^2)\, 
 \hat F_{\mu\nu}(p) \bigr] + \mathcal{O}(F^3)\,.
\end{align}

We emphasize now that it is not necessary to assume that the function $f$ is a Laplace transform of a function 
$\varphi$ which could be a distribution. This is a point where we can enjoy the beauty of the spectral action. 
Here, this assumption is an artifact due to the use of the heat kernel in formula \eqref{Trace}.

For all $k\in \set{1,\cdots,m}$, denote by $f^{[k]}(x) \vc -\int_x^\infty dy\,f^{[k-1]}(y)$ be a $k$-th primitive of $f$ (so 
we have $(f^{[k]})^{(k)}=f$).

\begin{lemma}
\label{intmult}
Assume $f\in \CC_2$. Then, for  $k\in \set{1,\cdots,m}$, $(-1)^k\,f^{[k]}$ is a bounded  
positive function in $C^k([0,\infty[)$ such that $(-1)^k\,f^{[k]}=\mathcal{O}\big((-1)^k\,f_\epsilon^{[k]}\big)$. Moreover
\begin{align*}
\int_0^\infty ds \int_{s}^\infty ds_1\int_{{s_1}} ds_2\cdots \int_{s_{k-2}}^\infty ds_{k-1}\,f(s_{k-1})
=(-1)^{k}\,f^{[k]}(0) \geq 0.
\end{align*}
\end{lemma}

\begin{proof}
Positivity and differentiability of $-f^{[1]}$ are clear. Let $x_0$ be such that $f(x) \leq c\,f_\epsilon(x)$, 
$\forall x>x_0$. Then, $0\leq -f^{[1]}(x)\leq c\int_x^\infty dy \,f_\epsilon^{[1]}(y)\leq-c\,f_\epsilon^{[1]}(x)$ and 
$-f^{[1]}=\mathcal{O}(-f_\epsilon^{[1]})$. An induction yields 
$0 \leq (-1)^k\,f^{[k]}=\mathcal{O}\big((-1)^k\,f_\epsilon^{[k]}\big)$. The integral identity follows from the fact that 
$f^{[k]}(x)\to 0$ when $x \to \infty$.
\end{proof}

Since we want to make a connection to \eqref{cf} and \eqref{wLambda}, we are obliged to suppose that $f$ is 
a Laplace transform but eventually of a distribution since, typically, we want to allow $f(x)$ to be $e^{-x}$. 

The next lemma shows $f\in \CC_2$ should be sufficient to get Theorem \ref{Result} at the price to 
redo all the covariant perturbation theory results of \cite{Barvinsky:1990up}, something that we do not do here. So we 
will assume that $f$ is a Laplace transform of a certain class of distributions defined on purpose.

It is worthwhile to quote at this point that there is a nice class of functions satisfying \eqref{hypothesis} (if one adjusts 
the tail of the function), namely the class of completely monotonic functions:

\begin{definition}
\label{defcm}
A function $f: \,]0,\infty[\to\mathbb{R}$ is completely monotonic (c. m.) if $f^{(n)}(x)$ exists for all $n\in \N$ and 
$(-1)^n\,f^{(n)}(x)\geq 0$ for all $x>0$.
\end{definition}
The limit $f^{(n)}(0)=\lim_{x\to 0^+} f^{(n)}(x)$ exists, finite or infinite. 
By Bernstein's theorem, a necessary and sufficient condition that $f$ should be c. m. in $]0,\infty[$ is that 
$f(x)=\int_0^\infty e^{-xt} \, dg(t)$ where $g(t)$ is a bounded non-decreasing function and that the integral 
converges for $x\in ]0,\infty[$ (see \cite[p. 161]{Widder}).
If $0\neq f$ is c. m., then $f(x)$ cannot vanish, and when $f^{(n)}(x_0)=0$ for $x_0>0$, then $f^{(m)}(x_0)=0$ 
for all $m>n$.

Examples of c. m. functions: 
$e^{-\a x},\,(\a+\beta x)^\ga, \, \log(a+\tfrac{b}{x})$ for $\a,\beta, \ga \geq 0$, $a\geq 1,b>0$.
This class is quite stable: if $f$ and $g$ are c. m. then, $f^{(2n)}(x)$, $-f^{(2n+1)(x)}$, $af(x)+bg(x)$ (for $a,b \geq 0$), 
$f(x)g(x)$, $e^{f(x)}$ are c. m. In particular $e^{ax^b},\,(1+x)^{a/x}, \,(\a+\tfrac{\beta}{x})^a$ are c. m. for 
$a\geq0,\,b\leq 0,\,\a \geq 1,\,\beta>0$.

Any Laplace transform $L(\varphi)(x) =\int_0^\infty dt\,e^{-xt}\,\varphi(t)$ of a non-negative function $\varphi$ is 
c. m. (if integral converges!). Any c. m. function $f$ is positive and Laplace transform of a positive function 
$\varphi$. Note that $f_\epsilon$ is c. m. and in $\kappa$ defined in \eqref{kappadef}, the key function $h$ 
of \eqref{h} is also c. m. and so is $-\kappa(x)=-2^m\, h'(x)$.
\medskip

To allow commutation of integrals in the proof of next lemma, we will use the hypothesis defined in Section
\ref{sec-ncom}.

\begin{lemma}
\label{lemmawL}
Assume that $f$ satisfies \eqref{hypothesis} and \eqref{LT}. Then
\begin{align}
&c_f=(-1)^{m}\,f^{[m]}(0),\label{cf1}\\
&w_\Lambda(p^2)=(-1)^{m-1}\,2^{m-2}\int_0^1d\a\,\Big[f^{[m-2]} \big(\a(1-\a)\Lambda^{-2}p^2\big)
-\tfrac{2\Lambda^{2}}{p^2} \int^{\a(1-\a)\Lambda^{-2}p^2}_0 ds_1\, f^{[m-2]}(s_1) \Big]. \label{w}
\end{align}
\end{lemma}

\begin{proof}
The first equation follows from 
$c_f=L\big(t^{-m}\varphi(t)\big)(0)$ and a property of Laplace transform: for 
$k\in \N^*$, 
$L\big(t^{-k}\varphi(t)\big)(s)=\int_s^\infty ds_1 \cdots \int_{s_{k-1}}^\infty ds_{k}\,f(s_{k})$, so that we can use 
Lemma \ref{intmult}. 
We apply again this formula to $ \kappa(z) =2^{m-2}\big(-h(z) -2\,\tfrac{h(z)-1}{z}\big)$ in the function 
$w_\Lambda(\Lambda^2\,x)=\int_0^\infty dt\, \kappa(t\,x)\,\varphi(t)\,t^{-(m-2)}$: for instance,
\begin{align}
\int_0^\infty dt \, h(t\,x)\,\varphi(t)t^{-(m-2)} & = \int_0^\infty dt\,\int_0^1d\a\, e^{-\a(1-\a)tx}\,
\varphi(t)\,t^{-(m-2)} \nonumber\\
&= \int_0^1d\a\, \int_0^\infty dt \, e^{-\big(\a(1-\a)x\big)t}\,\varphi(t)\,
t^{-(m-2)} \nonumber\\
& = \int_0^1 d\a \, L\big(\varphi(t)\, t^{-(m-2)})(\a(1-\a)x) \nonumber \\
& = \int_0^1d\a \, \int_{\a(1-\a)x}^\infty ds_1\int_{s_1}^\infty ds_2 \cdots 
\int_{s_{m-3}}^\infty ds_{m-2}\, f(s_{m-2}) \nonumber\\
& = \int_0^1d\a \,(-1)^{m-2} f^{[m-2]}\big(\a(1-\a)x\big). \label{inth}
\end{align}
The commutation of integrals in the second equality follows from \eqref{LT}.
To complete the proof, we proceed similarly for
\begin{align*}
\int_0^\infty dt \,\tfrac{h(t\,x)-1}{t\,x}\,\tfrac{\varphi(f)(t)}{t^{m-2}}  &= \tfrac{1}{x}\int_0^\infty dt\, h(t\,x)\,
\tfrac{\varphi(f)(t)}{t^{m-1}}-\tfrac{1}{x}\int_0^\infty dt\,\tfrac{\varphi(f)(t)}{t^{m-1}} \\
& = \tfrac{1}{x}\int_0^1d\a \,\big[ L\big(\varphi(t)t^{1-m-1}\big) \big(\a(1-\a)x\big)
-L\big(\varphi(t)t^{1-m}\big)(0) \big] \\
& =  -\tfrac{1}{x} \int_0^1 d\a \, \int_0^{\a(1-\a)x} ds_1 \int_{s_1}^\infty ds_2 \cdots 
\int_{s_{m-2}}^\infty ds_{m-1}\,f(s_{m-1}) \\
& =  \tfrac{(-1)^{m-1}}{x} \int_0^1 d\a \, \int_0^{\a(1-\a)x} ds_1\,f^{[m-2]}(s_1).
\end{align*}
\end{proof}
For instance, in dimension $d=4$, formulae \eqref{cf1} and \eqref{w} read
\begin{align}
&c_f=\int_0^\infty ds\,\int_{s}^\infty ds_1\, f(s_1),\nonumber \\
&w_\Lambda(p^2)=-\int_0^1d\a\,\Big[f \big(\a(1-\a)\Lambda^{-2}p^2\big)
-\tfrac{2\Lambda^{2}}{p^{2}} \int^{\a(1-\a)\Lambda^{-2}p^2}_0 ds_1\, f(s_1) \Big]. 
\label{w4d}
\end{align}
Finally, we get the following

\begin{theorem}
\label{Result}
Assume that $f\in \CC_1 \cap \CC_2 \cap \CC_3$. Then, up to the 
second order in $F$, 
the variation of the spectral action of above spectral triple is given by \eqref{Squadra} where $c_f$ and $w$ are 
defined by \eqref{cf1} and \eqref{w}. Moreover the asymptotics of $w$ is controlled by
\begin{align}
\label{wlarge}
w_\Lambda(p^2) \, \underset{p^2\to \infty}{\sim} \, -2^{m+1}\,c_f \,\Lambda^4\, p^{-4}+ o(p^{-4}).
\end{align}
\end{theorem}

\begin{proof}
It remains to prove the asymptotics. 
Since $\kappa(x)=2^m\,\tfrac{d\,}{dx}h(x)$,
$\kappa(t\Lambda^{-2}z)=\tfrac{2^m \Lambda^2}{t}\tfrac{d\,}{dz} h(t\Lambda^{-2}z)$,
with $z\vc p^2$,
\begin{align}
w_\Lambda(z)=&2^m\Lambda^{2} \tfrac{d\,}{dz}\int_0^\infty dt\,t^{-m+1}h(t\Lambda^{-2}z)\, 
\varphi(t)\nonumber\\
=&2^m\Lambda^{2} \tfrac{d\,}{dz} \int_0^1d\a\,(-1)^{m-1}\,f^{[m-1]}\big(\a(1-\a)z\Lambda^{-2}\big) 
\label{wddz}
\end{align}
by \eqref{inth}.
Moreover
\begin{align*}
\int_0^1d\a\,(-1)^{m-1}\,f^{[m-1]}\big(\a(1-\a)z\Lambda^{-2}\big)& \, \underset{x\to \infty}{\sim} \, 
(-1)^{m-1}2\int_0^{1/2}d\a\,  f^{[m-1]}(\a z\Lambda^{-2}\big) \\
& \, \underset{z\to \infty}{\sim} \, \tfrac{(-1)^{m-1}2\Lambda^2}{z}\int_0^\infty dy\,f^{[m-1]}(y)
\\ &\,\underset{\,\,\,\,\,\,\,\,\,\,\,\,\,\,\,\,\,\,\,}{=\,\,\,} \tfrac{-2(-1)^{m-1}\Lambda^2}{z}\,f^{[m]}(0)=\tfrac{2\Lambda^2\,\,c_f}{z}\,.
\end{align*}
Thus
\begin{align*}
w_\Lambda(z)  \, \underset{z\to \infty}{\sim} \, 
2^{m+1}\Lambda^4\,c_f \tfrac{4\,}{dz}\Big( \tfrac{1}{z}\Big).
\end{align*}
\end{proof}

\begin{remark}\label{dual}
The coefficient $\Lambda^4c_f$ is the same one which appears in front of
the $a_0$ term in the weak-field asymptotics of the spectral action. Although in our case, this coefficient is canceled 
due to the subtraction of the free heat kernel, cf.\ Eq.\ (\ref{asS}), it contributes to the cosmological constant term 
for more general spectral triples allowing for metric perturbations \cite{ConnesMarcolli}. This is a very remarkable
fact which shows an intimate relation between leading terms in the low and high momenta 
asymptotics of the spectral action. 
A positive value for $c_f$ is required to reproduce the Standard Model from noncommutative geometry 
\cite{CCM,ConnesMarcolli} (see also discussion in Sec.\ \ref{sec-pos} below).
\end{remark}

\begin{remark}
While the spectral action is non-local (a price to pay for its extremely simple definition), it becomes local after some 
asymptotic expansion. But the covariant perturbation method used from \cite{Barvinsky:1990up} is a way to control 
non-localities.
\end{remark}

\begin{remark}
To get the theorem, a spin structure is not necessary on $M$: any second order operator of the Laplace type with 
values on a vector bundle $V$ over $M$ (see \cite[eq. (2.1)]{Vassilevich:2003xt}) can replace $\DD^2$.
\end{remark}

\begin{remark}
By Theorem \ref{Result},the function $w_\Lambda(p^2)$ decays  at large $p$ at least as $1/p^4$. 
This implies that the propagator for the Yang-Mills field grows at large momenta, and the Yang-Mills spectral 
action is not super-renormalizable in contrast to the restricted expansion of the spectral action studied in 
\cite{Suijlekom2}.
\end{remark}

\subsection{An example}

Assume that $d=4$, and as a function $f$, let us take the function $f_r(z) \vc (z+a)^{-r}$ having a 
power-low decay at infinity where $0\neq a\in \R^+$ and $r>2$, so $f\in\CC_1$. 
To estimate the asymptotic behavior of spectral action it is most convenient to use \eqref{wddz}. Since the first 
primitive of $f$ is $f^{[1]}_r=(1-r)^{-1} (z+a)^{-r+1}$, we have
\begin{align*}
w_{\Lambda,r}(z)&=4\Lambda^2 \tfrac{d\ }{dz} 
\int_0^1 d\alpha\, \tfrac{-1}{1-r} [\alpha(1-\alpha)z\Lambda^{-2} +a]^{-r+1}
=\tfrac 8{r-1} \tfrac{d\ }{dy} \int_0^{1/2} d\alpha\, [\alpha(1-\alpha)y +a]^{-r+1}\\
&=\tfrac 8{r-1} \tfrac{d\ }{dy} y^{-r+1}\int_0^{1/2} d\alpha\, \Bigl[\alpha(1-\alpha) 
+\tfrac ay \Bigr]^{-r+1}
\end{align*}
with $y=z\Lambda^{-2}=p^2\Lambda^{-2}$. This expression is valid for arbitrary $p^2$. 

To evaluate the asymptotic behavior at $p^2\to \infty$ (or, at $y\to\infty$), one represents the integrand
$(a/y)+\alpha (1-\alpha)=(a/y)+\alpha-\alpha^2$ and expands in $\alpha^2$: 
\begin{equation*}
w_{\Lambda,r}(p^2)= \tfrac 8{r-1} \tfrac{d\ }{dy} y^{-r+1}\int_0^{1/2} d\alpha\,
\Bigl[ (\alpha + (a/y))^{-r+1} -(1-r)\alpha^2(\alpha + (a/y))^{-r}+\dots \Bigr]
\end{equation*}
yielding for the leading asymptotics
\begin{equation*}
w_{\Lambda,r}(p^2)=-\tfrac {8}{(4\pi)^2  a^{r-2}(r-1)(r-2)}\, 
\tfrac{\Lambda^{4}}{p^4}+o(p^{-4})
\end{equation*}
as claimed in 
\eqref{wlarge} since $c_f=\tfrac{1}{a^{r-2}(r-1)(r-2)}$. 

It is interesting to estimate the next-to-leading term in the expansion, which behaves as
$$\begin{array}{lll}
&p^{-6} \,&\text{ for}\quad r>3,\\
&p^{-6} \ln {p}  \,&\text{ for}\quad r=3,\\
&p^{-2r} \,&\text{ for}\quad 2<r<3.
\end{array}
$$
This example show that the next to leading term in the asymptotics is not always $p^{-6}$ as one might think.

\section{\texorpdfstring{Relaxing assumptions on the function $f$}{Relaxing assumptions on the function f}}
\label{sec-ass}

In the previous section, we used different sets of assumptions on the function $f$ which defines the spectral 
action. The first set $\CC_1$ refers to the role of $f$ as a regulator of the trace while $\CC_2$ allows to take 
primitives of $f$. It is natural to assume that $f$ is positive, and it is necessary to assume that it falls off 
sufficiently fast. The third set $\CC_3$ is of a technical nature: we need $f$ being a Laplace transform of a 
function $\varphi$ to use the heat kernel methods. The form of $\varphi$ was restricted
to enable us to commute integrals in the proof of Theorem \ref{intmult}.

One can ask whether there is a natural family of functions satisfying all assumptions which we have made. 
Such a family consists of the so-called completely monotonic functions and is described in
the Definition \ref{defcm}. These functions seem to be the best candidates for the definition of spectral action.

Another question, which shall be address below, is whether some of the assumptions can be weakened.

\subsection{Lifting the positivity assumption}\label{sec-pos}

For a positive function $f$, the $k$-th primitive $f^{[k]}$ is positive for $k$ even and negative for $k$ odd. 
Then, according to (\ref{cf1}) the coefficient $c_f$ is strictly positive. The leading $1/p^4$ term in the 
expansion (\ref{wlarge}) is then non-zero. However, lifting the positivity assumption on $f$, one gets a possibility
to engineer a variation of spectral action such that $w_\Lambda(p^2)$ falls off at large $p$ as $p^{-2k}$ for 
any given integer $k>2$.  Indeed, consider $f$ being a finite sum
\begin{align*}
f(z)=\sum_n f_n(z), \qquad f_n(z)=c_n\,\exp (-z/b_n) 
\end{align*}
with $b_n>0$ and where some of $c_n$ are allowed to take negative values. Recall that 
$c_n^{-1}f_n\in\CC_1\cap\CC_2\cap\CC_3$.

The variation of spectral action is then given by a sum of the heat kernels
\begin{align}
\label{SAfn}
\SS(\DD_A,\Lambda,f)= \sum_n c_n\, \tilde K(b_n^{-1}\Lambda^{-2},\DD_A^2)\,,
\end{align}
and the function $w_\Lambda(p^2)$ can be calculated directly without using
the Laplace transform:
\begin{align*}
w_\Lambda (p^2) = 4\Lambda^2 \tfrac{d\ }{d(p^2)} \sum_n c_n\,b_n\, h(p^2\Lambda^{-2}b_n^{-1})\,. 
\end{align*}
Here, for simplicity, we restrict ourselves to four dimensions, so $m=2$.

Using \eqref{hzas1}, the large $p$-asymptotic expansion for $w_\Lambda(p^2)$ reads
\begin{equation}
\label{wLbc}
w_\Lambda(p^2) \,\underset{p\to \infty}{\sim}\, -8 \sum_{j=0}^\infty (\Lambda/p)^{2(j+2)}\,
\tfrac{(j+1)(2j)!}{j!} \sum_{n} c_n \,b_n^{2+j} \,.
\end{equation}

By comparing (\ref{SAfn}) with (\ref{asS}), we see that in the weak field asymptotics of the action, 
the heat kernel coefficients $a_{2k}$ are multiplied by a sum $\sum c_n \,b^{2-k}$ having the same 
structure as the sums appearing in the large $p$-asymptotics (\ref{wLbc}) but with different 
powers of $b_n$. The only exception is $\sum c_n \,b_n^2$ which multiplies the leading $p^{-4}$ 
term in (\ref{wLbc}) and the coefficient $a_0$. (We remind, that $a_0$ reappears in the asymptotics of the
spectral action if we allow for metric perturbations.) Herewith we reconfirm the
statement made in Remark \ref{dual}, but also observe that no further relations
exist between the weak-field and large momentum asymptotics of the spectral action. 

If one is allowed to choose freely the coefficients $b_n$ and $c_n$ and to take negative $c_n$, it is possible 
to cancel any finite number of terms in (\ref{wLbc}).

\subsection{Step-function regularization}
\label{sec-step}

The Laplace transformation is smoothing: if $f$ is a Laplace transform of $\varphi$, then $f$ has to be at 
least continuous. However, the assumptions (\ref{hypothesis}) remain valid also for piecewise 
continuous functions, and the right hand sides of the non-asymptotic (\ref{w}) and asymptotic (\ref{wlarge})
formulae exist for such functions. It is therefore interesting to make a calculation for a piecewise 
continuous function $f$ without using the Laplace transformation. Such case has been investigated in \cite{EGBV}: 
using the spectral density of the operator $\DD$, a proof that the asymptotics of spectral action makes sense in 
the Ces\`aro sense, has been proposed when $M$ is compact.

Following \cite{Andrianov:2010nr}, let us take $f$ being a step-function and use the following integral representation.
\begin{align*}
f(z)=\theta (1-z) = \lim_{\epsilon\to +0} \tfrac 1{2\pi i}
\int_{-\infty}^{\infty} d\beta \,\tfrac{e^{i(1-z)\beta}}{\beta -i\epsilon},
\end{align*}
so that
\begin{equation*}
\Tr\big( f(\DD_A^2/\Lambda^2)-f(\DD^2/\Lambda^2)\big) =\lim_{\epsilon\to +0} \,
\tfrac 1{2\pi i} \int_{-\infty}^{\infty} d\beta \,\tfrac{e^{i\beta}}{\beta -i\epsilon}\,
\Tr \Bigl[ \exp \Bigl( -\tfrac {i\beta}{\Lambda^2} \DD_A^2 \Bigr)-\exp \Bigl( -\tfrac {i\beta}{\Lambda^2} \DD^2 \Bigr)\Bigr].
\end{equation*}

The coefficient in front of $\DD_A^2$ above is imaginary, so that we are dealing with the Schr\"odinger 
kernel rather than with the heat kernel. In principle, one should study anew the expansion of this kernel
in curvatures. We shall not go too deep into this problem. It is enough to mention that (i) the heat kernel 
for the free operator $\DD^2$ becomes the Schr\"odinger kernel for the same operator after analytic
continuation $t\to it$, and (ii) the Duhamel expansion, which may be used to derive the formulae like 
\eqref{traceheat} (see, e.g. \cite[Sec.\ 8.2]{Vassilevich:2003xt}) also survives the continuation to imaginary $t$.

In four dimensions, $m=2$, we therefore write
\begin{align*}
\Tr\Bigl[ \exp \Bigl( -\tfrac {i\beta}{\Lambda^2} \DD_A^2 \Bigr) -\exp \Bigl( -\tfrac {i\beta}{\Lambda^2} \DD^2 \Bigr)\Bigr]
\sim \tfrac 1{(4\pi)^2} \int d^4 p\, \tr\, 
\bigl[\hat F^{\mu\nu}(-p)\,\kappa(i\beta p^2 \Lambda^{-2})\, \hat F_{\mu\nu}(p)\bigr]
\end{align*}
to the second order of the field strength $F$. The formula (\ref{Squadra}) remains valid with  
$w_\Lambda(p^2)$ given by
\begin{align*}
w_\Lambda(p^2)=\lim_{\epsilon\to +0} \,\tfrac 1{2\pi i} \int_{-\infty}^{\infty} d\beta \,\tfrac{e^{i\beta}}{\beta -i\epsilon}
\,\kappa(i\beta p^2 \Lambda^{-2}) 
\end{align*}
According to (\ref{kappadef}), the function $\kappa(z)$ contains two terms. Let us denote by $w_{1,\Lambda}$ 
(resp., $w_{2,\Lambda}$) the contribution from $-h(z)$ (respectively, from $4q(z)$). 
\begin{eqnarray*}
 w_{1,\Lambda}&=&-\lim_{\epsilon\to +0} \,\tfrac 1{2\pi i}
\int_{-\infty}^{\infty} d\beta \, \tfrac{e^{i\beta}}{\beta -i\epsilon}\,h(i\beta p^2 \Lambda^{-2})\\
&=&-\lim_{\epsilon\to +0} \,\tfrac 1{2\pi i}
\int_{-\infty}^{\infty} d\beta\int_0^1 d\alpha\, \tfrac{e^{i\beta}}{\beta -i\epsilon}
\exp\big(-i\beta \alpha(1-\alpha) p^2 \Lambda^{-2}\big)=-\int_0^1d\alpha \, \theta_\alpha\,,
\end{eqnarray*}
where we introduced a short-hand notation
\begin{equation*}
 \theta_\alpha :=\theta \bigl( 1 -\alpha(1-\alpha) p^2 \Lambda^{-2}\bigr)\,.
\end{equation*}
To evaluate $w_{2,\Lambda}$, we first observe that $q(z)$ does not have a pole at $z=0$. Therefore, one can 
deform the contour of integration by replacing the interval $[-\epsilon/2,\epsilon/2]$ by a semicircle in the 
upper half plane centered at $\beta=0$ of a radius $\epsilon/2$. We denote the deformed contour
by $C$. 
Then we have
\begin{eqnarray}
&&w_{2,\Lambda}=\lim_{\epsilon\to +0} \,\tfrac 1{2\pi i}
\int_C d\beta \, \tfrac{e^{i\beta}}{\beta -i\epsilon} \, 
\tfrac{2(1-h(i\beta p^2\Lambda^{-2}))}{i\beta p^2\Lambda^{-2}}\nonumber\\
&&\qquad = \lim_{\epsilon\to +0} \,\tfrac 1{2\pi i}
\int_C d\beta \, \tfrac{e^{i\beta}}{\beta -i\epsilon} \int_0^1 d\alpha 
\tfrac{2}{i\beta p^2\Lambda^{-2}} \Bigl[ 1 -\exp (-i\beta \alpha (1-\alpha)p^2\Lambda^{-2})
\Bigr]\,. \nonumber
\end{eqnarray}
Next, we integrate over $\beta$. In the integral of the first (constant) term in the square brackets, the contour 
should be closed in the upper half plane, so that the pole at $\beta=i\epsilon$ contributes. To integrate the 
second term  in square brackets, we close the contour upwards for positive arguments of $\theta_\alpha$, 
obtaining a contribution proportional to the residue at $\beta=i \epsilon$. For negative arguments of $\theta_\alpha$, 
the contour has to be closed in the lower half-plane, giving rise to a contribution proportional to $(1-\theta_\alpha)$ 
and to a residue at $\beta=0$. Collecting all together, we obtain
\begin{equation*}
w_{2,\Lambda}=\int_0^1d\alpha\, \lim_{\epsilon\to +0} \Bigl( -\tfrac{2e^{-\epsilon}}{\epsilon p^2\Lambda^{-2}}
+\theta_\alpha \tfrac{2e^{-\epsilon}}{\epsilon p^2\Lambda^{-2}} 
\exp( \epsilon \alpha (1-\alpha) p^2\Lambda^{-2}) 
-(1-\theta_\alpha) \tfrac{-2}{\epsilon p^2\Lambda^{-2}} \Bigr)\,.
\end{equation*}
In the limit $\epsilon\to 0$ we have
\begin{align*}
w_{2,\Lambda}=\int_0^1d\alpha\, \bigl[ 2\alpha (1-\alpha) \theta_\alpha 
+2\Lambda^2 p^{-2} (1-\theta_\alpha)\bigr]\,.
\end{align*}
One can check that $w_{\Lambda}=w_{1,\Lambda}+w_{2,\Lambda}$ agrees with (\ref{w4d}). The integration over 
$\alpha$ can be performed giving
\begin{eqnarray*}
&2|p|<\Lambda:&\quad w_{\Lambda}(p^2)=-\tfrac 23 \\
&2|p|\ge \Lambda:&\quad w_{\Lambda}(p^2)= 2\bigl[ -\alpha_c +\tfrac {2\Lambda^2}{p^2}
\bigl( \tfrac 12 -\alpha_c\bigr)+\alpha_c^2 -\tfrac 23 \alpha_c^3 \,\bigr]\,,
\end{eqnarray*}
where $\alpha_c:=\tfrac 12 \big( 1-\sqrt{1-\tfrac{4\Lambda^2}{p^2}} \,\,\big)$.

Also the large $p$ asymptotic behavior $w_\Lambda(p^2)\sim -4\Lambda^4/p^4 +\mathcal{O}(1/p^6)$ agrees with 
\eqref{wlarge} although the step-function does not satisfy assumptions of Theorem \ref{Result}.

We conclude that the Laplace transform is not strictly necessary. We can even conjecture that the formulae \eqref{w} 
and (\ref{wlarge}) are valid for arbitrary piecewise continuous functions $f$ satisfying the fall-off condition 
(\ref{hypothesis}).

\section{\texorpdfstring{Higher terms in $F$}{Higher terms in F}}
\label{sec-ht}

Since the heat kernel \eqref{traceheat} is an invariant functional in the background field, it is expandable in the 
basis of invariants of any order. This has been explicitly computed in \cite{BGVZ,BGVZ1} up to order 3 in 
the curvature $F$:
\begin{align}
\tilde K(s,\DD_A^2)(F)=\tfrac{1}{(4\pi s )^m}  \int_M d^{2m}x \
\tr\, \Big\{sE&+ s^2\,E\,\tfrac{1}{2}h(-s\partial^2)\,E  + s^2\,F_{\mu\nu}\, q(-s\partial^2) \,F^{\mu\nu} \nonumber   \\
& + s^3 \sum_{i=1}^{11} \kappa_i(-s\partial^2_1,-s\partial^2_2,-s\partial^2_3)(R_1R_2R_3)(i)\nonumber \\
&+s^4 \sum_{i=12}^{25} \kappa_i(-s\partial^2_1,-s\partial^2_2,-s\partial^2_3)(R_1R_2R_3)(i)\label{3rd}\\
&+s^5 \sum_{i=26}^{28} \kappa_i(-s\partial^2_1,-s\partial^2_2,-s\partial^2_3)(R_1R_2R_3)(i)\nonumber\\
&+s^6 \,\,\kappa_{29}(-s\partial^2_1,-s\partial^2_2,-s\partial^2_3)(R_1R_2R_3)(29 )\quad 
+\mathcal{O}\big(F^4\big)(s) \Big\} \nonumber
\end{align}
where the list of functions $\kappa_i$ is known: they are constructed like previous 
$\kappa_0(x)=2^m h'(x)$ through the function 
\begin{equation*}
h(x_1,x_2,x_3) \vc \int_{(R^+)^3}d\a_1\,d\a_2\,d\a_3\,\delta(1-\a_1-\a_2-\a_3)\,
e^{-\a_1\a_2x_3-\a_2\a_3x_1-\a_1\a_3x_2}.
\end{equation*} 
The formal expression $R_1R_2R_3$ means one of the cross-terms in $F$ and other curvatures that can appear 
in the computation.

Exactly the same method as before can be applied. While $w_\Lambda(p^2)$ defines behavior of 
the propagator and partially of the $A^3$ and $A^4$ vertexes, the action (\ref{3rd}) will define higher vertex 
functions.

\section{Conclusions}
\label{sec-con}

In this paper we considered the variation of spectral action for a commutative spectral triple perturbed by a gauge 
potential. We calculated this variation to the second order in field strength. Our main results are the 
remarkably simple formula (\ref{w}) for the action and its universal asymptotics, Theorem \ref{Result}. 

Although we used the Laplace transform at intermediate steps, it seems to be unnecessary, as suggests the  
example of step function regularization. Anyway, it would be interesting to obtain the results 
without relying on the Laplace transform. Other open directions for further research include
extensions to odd dimensions and to higher orders in the field strength.

\section*{Acknowledgments}

This work was started during the visit of the first author (B. I.) to S\~ao Paulo which was supported by FAPESP. 
The second author (C. L.) was supported by the DNRF through the Centre for Symmetry and Deformation in 
Copenhagen. The third author (D. V.) was supported in part by CNPq and FAPESP.


\begin{thebibliography}{29}

\bibitem{Andrianov:2010nr}
A.~A.~Andrianov, F.~Lizzi,
``Bosonic Spectral Action Induced from Anomaly Cancellation'', 
JHEP {\bf 1005}, (2010)  057.

\bibitem{Barvinsky:1987uw}
A.~O.~Barvinsky, G.~A.~Vilkovisky,
``Beyond the Schwinger-Dewitt Technique: Converting Loops Into Trees and In-In Currents'', 
Nucl.\ Phys.\  {\bf B282} (1987), 163--188.
  
\bibitem{Barvinsky:1990up}
A.~O.~Barvinsky, G.~A.~Vilkovisky, 
``Covariant perturbation theory (II): Second order in the curvature. General algorithms'', 
Nucl.\ Phys.\  {\bf B333} (1990), 471--511.

\bibitem{BGVZ}
A.~O.~Barvinsky, Yu V. Gusev, G.~A.~Vilkovisky and V. V. Zhytnikov, 
``The basis of nonlocal curvature invariants in quantum gravity theory. (Third order)'', 
J. Math. Phys. {\bf 35} (1994), 3525--3542.

\bibitem{BGVZ1}
A.~O.~Barvinsky, Yu V. Gusev, G.~A.~Vilkovisky and V. V. Zhytnikov, ``Asymptotic behaviors of the heat kernel 
in covariant perturbation theory'', 
J. Math. Phys. {\bf 35} (1994), 3543--3559.

\bibitem{BuzanoNicola}
E. Buzano and F. Nicola, ``Pseudo-differential operators and
Schatten-von Neumann classes'', Advances in Pseudo-Differential Operators (Boggiatto P.,
Ashino R. and Wong M. W., eds.), Operator Theory: Advances and Applications
{\bf 155}, Birkh\"{a}user  2004, 117--130.


\bibitem{BuzanoToft}
E. Buzano and J. Toft, ``Schatten-von Neumann properties in the Weyl calculus'', 
J. Funct. Anal., {\bf 259} (2012), 3080--3114.

\bibitem{CGRS}
A. Carey, V. Gayral, A. Rennie and F. Sukochev, 
``Integration on locally compact noncommutative spaces'', 
arXiv:0912.2817v1 [math.OA].

\bibitem{CC}
A. Chamseddine and A. Connes,
``The spectral action principle'',
Commun. Math. Phys. {\bf 186} (1997), 731--750.

\bibitem{CC1}
A. Chamseddine and A. Connes,
``Inner fluctuations of the spectral action'',
J. Geom. Phys. {\bf 57} (2006), 1--21.

\bibitem{CC2}
A. Chamseddine and A. Connes, 
``The uncanny precision of the spectral action'',
Commun. Math. Phys. {\bf 293} (2010), 867--897.

\bibitem{CCM}
A. Chamseddine, A. Connes and M. Marcolli, 
``Gravity and the standard model with neutrino mixing'', 
Adv. Theor. Math. Phys. {\bf 11} (2007), 991--1090.

\bibitem{Book}
A. Connes,
\emph{Noncommutative Geometry},
Academic Press, London and San Diego, 1994.

\bibitem{ConnesMarcolli}
A. Connes and M. Marcolli,
\emph{Noncommutative Geometry, Quantum Fields and Motives},
Colloquium Publications, Vol. 55, American Mathematical Society, 2008.

\bibitem{EGBV}
R. Estrada, J. M. Gracia-Bond\'{\i}a and J. C. V\'arilly, 
``On summability of distributions and spectral geometry'', 
Commun. Math. Phys. {\bf 191} (1998), 219--248.

\bibitem{GGBISV}
V. Gayral, J. M. Gracia-Bond\'{\i}a, B. Iochum, T. Sch\"ucker and J. C. V\'arilly, 
``Moyal planes are spectral triples'', 
Commun. Math. Phys. {\bf 246} (2004), 569--623.

\bibitem{GI}
V. Gayral, B. Iochum, 
The spectral action for Moyal planes, 
J. Math. Phys. {\bf 46} (2005), 043503.

\bibitem{GilkeyNew}
P.~B.~Gilkey, \emph{Asymptotic Formulae in Spectral Geometry},
Chapman \& Hall/CRC, Boca Raton, 2004.

\bibitem{Polaris}
J. M. Gracia-Bond\'{\i}a, J. C. V\'arilly and H. Figueroa,
\emph{Elements of Noncommutative Geometry},
Birkh\"auser Advanced Texts, Birkh\"auser, Boston, 2001.

\bibitem{NR}
F. Nicola and L. Rodino, \emph{Global Pseudo-Differential Calculus on Euclidean Spaces}, 
Springer Basel AG, Birk\"{a}user 2010.

\bibitem{Pushnitski09}
A. Pushnitski, 
``Spectral theory of discontinuous functions of selfadjoint operators: essential spectrum'', 
Integr. Equ. Oper. Theory {\bf 68} (2010), 75--99.

\bibitem{Pushnitski10}
A. Pushnitski, 
``An integer-valued version of the Birman-Krein formula'', 
arXiv:1006.0639 [math.SP].

\bibitem{Simon}
B. Simon, 
\emph{Trace Ideals and Their Applications}, 
Second Edition, AMS Mathematical Surveys and Monographs 120, Providence, 2005.

\bibitem{Suijlekom2}
W.~D.~van Suijlekom,
``Renormalization of the asymptotically expanded Yang-Mills spectral action', 
Commun. Math. Phys. {\bf 312} (2012), 883--912.

\bibitem{Vassilevich:2003xt}
D.~V.~Vassilevich,
``Heat kernel expansion: User's manual,''
Phys.\ Rept.\  {\bf 388} (2003), 279-360. [hep-th/0306138].

\bibitem{Widder}
D. Widder, 
\emph{The Laplace Transform}, 
Princeton University Press, Princeton, 1946.

\bibitem{Yafaev92}
D. R. Yafaev, 
\emph{Mathematical Scattering Theory: General Theory}, Translations of Mathematical Monographs 105, 
Amer.Math.Soc., Providence, (1992).

\bibitem{Yafaev05}
D. R. Yafaev, 
``A trace formula for the Dirac operator'', 
Bull. London Math. Soc. {\bf 37} (2005), 908--918.

\bibitem{Yafaev07}
D. R. Yafaev, 
``The Schr\"odinger operator: perturbation determinants, the spectral shift function, trace identities, and all that'', 
Funct. Anal. and its Appl., {\bf 41} (2007), 217--236.


\end{thebibliography}
\end{document}